\documentclass[oneside,english]{amsart}
\usepackage[T1]{fontenc}
\usepackage[utf8]{inputenc}
\usepackage{amsthm}
\usepackage{geometry}
\makeatletter
\numberwithin{equation}{section}
\numberwithin{figure}{section}

\makeatother
\theoremstyle{plain}
\newtheorem{thm}{\protect\theoremname}
\newtheorem{lem}[thm]{\protect\lemmaname}
\newtheorem{prop}[thm]{\protect\propositionname}
\usepackage{babel}
\providecommand{\lemmaname}{Lemma}
\providecommand{\propositionname}{Proposition}
\providecommand{\theoremname}{Theorem}

\begin{document}

\title{Scaling Symmetry in Symplectic Thermodynamics}

\author{Mario C. Baldiotti 
 $^{1}$ and Rodrigo Fresneda $^{2,}$*}

\address{$^{1}$ \quad Department of Physics, Londrina State University, Londrina, Brazil, 86.057-970; baldiotti@uel.br
\\
$^{2}$ \quad Center for Mathematics, Computing and Cognition, Federal University of ABC, Santo André, Brazil, 09280-560}

\begin{abstract}This paper investigates scaling symmetry in thermodynamics by unifying constrained Hamiltonian dynamics with symplectic and contact geometries. Through the mathematical processes of contactization and symplectization, we demonstrate that fixing an extended global scale variable effectively recovers the standard thermodynamic description in terms of scale-invariant quantities. The geometric formalism is illustrated by establishing the diffeomorphism between the Lagrangian submanifolds of ideal and van der Waals gases. Finally, applying this framework to a Schwarzschild black hole reveals that {changing the scaling weights of entropy and internal energy} is a fundamental physical requirement to accommodate non-isothermal dynamics.\end{abstract}

\maketitle


\section{Introduction}

The geometric formulation of {equilibrium} thermodynamics has conventionally been grounded in the framework of contact manifolds~\cite{brav1}. However, alternative perspectives based on symplectic geometry and constrained Hamiltonian mechanics offer powerful tools for describing equilibrium systems. In~earlier works~\cite{baldiotti2016,baldiotti2017}, it was established that thermodynamic equations of state can be successfully modeled as primary constraints in a phase space, defining a constraint surface on which the tautological one-form equates to the differential of the internal energy. Because~any curve connecting two equilibrium states represents a valid idealized quasi-static process, the~system lacks physical degrees of freedom. Within~this symplectic perspective, an~equilibrium thermodynamic system resides on an $n$-dimensional Lagrangian submanifold embedded in a $2n$-dimensional symplectic phase~space.

 {The role of symmetries in the symplectic structure of thermodynamics, as~well as the relation of these symmetries to gauge transformations, have been considered in many contexts~\cite{balian2001,bravetti2015,bravetti2015b,quevedo2022,bravetti2023,gosh2026}}. This paper aims to investigate gauge fixing in thermodynamics from the standpoint of a symplectic framework based on the work in {\cite{baldiotti2016}} and to connect it with the standard contact-geometric approach. To~bridge these formalisms, we detail the mathematical operations of contactization and symplectization. Through contactization, the~$2n$-dimensional symplectic manifold is lifted to a $\left(2n+1\right)$-dimensional contact manifold, wherein the Lagrangian submanifolds are mapped to Legendre submanifolds that naturally encode the fundamental thermodynamic relations. Conversely, the~symplectization of the contact manifold yields a $\left(2n+2\right)$-dimensional symplectic space by formally introducing an extended global scale~variable.

To illustrate the validity and physical scope of this geometric development, we apply the formalism to classical thermodynamic systems, explicitly demonstrating the diffeomorphisms between the Lagrangian submanifolds of the ideal and van der Waals gases. Within~the symplectization framework, we demonstrate how a thermodynamic system originally described by scale-invariant quantities in a $\left(2n+1\right)$-dimensional contact manifold can be formally embedded into an extended symplectic manifold. This is achieved by introducing a global scale variable and its conjugate, yielding a description in terms of homogeneous quantities of arbitrary degree. By~applying a gauge-fixing procedure---specifically, holding this scale variable constant---we effectively take the quotient of the extended symplectic manifold by the orbits of the associated dilation generator, thereby recovering the original contact geometry and the scale-invariant thermodynamic formulation. Furthermore, we construct a Hamiltonian flow that generates scale-invariant quasi-static processes, which naturally project down to the equilibrium Legendre submanifold. This extended formalism is subsequently applied to the thermodynamics of a Schwarzschild black hole, illustrating that {changing the energy and entropy scaling weights} is geometrically and physically necessary to accommodate non-isothermal~dynamics.

First, we review the authors' earlier work on the Hamiltonian formalism and~then connect that development with the standard contact-geometric approach. In~Section~\ref{sec:Symplectic-approach},~we establish the connection between our earlier work~\cite{baldiotti2016} and symplectic geometry. In~Section~\ref{sec:Contactization}, we show how to transition from the symplectic description to the contact description while preserving the thermodynamic formulation. 
In~both approaches, we consider the ideal and van der Waals gases. In~the first case, we show there is a symplectomorphism connecting the two physical systems, while in the second case, there is a contactomorphism between them. Finally, in~Section~\ref{sec:Symplectization} we define symplectization and interpret the extra dimension as a global scaling gauge freedom. We apply the formalism to the description of the Schwarzschild black~hole.

{
\subsection*{Notation and Conventions}

For clarity and ease of reading, we summarize here the principal mathematical symbols and conventions used throughout this manuscript to describe the symplectic and contact geometric formulations of thermodynamics:

\paragraph{Spaces and Submanifolds}
\begin{itemize}
	\item $M$: The symplectic manifold.
	\item $\mathcal{C}$: The constraint submanifold.
	\item $C$: The contact manifold.
	\item $L$: The Lagrangian submanifold of $M$.
	\item $\Lambda$: The Legendre submanifold of $C$.
\end{itemize}

\paragraph{Coordinates and Variables}
\begin{itemize}
	\item $q^i$: The scale-invariant densities of extensive thermodynamic variables (e.g., volume, particle number).
	\item {$p_i$: The intensive thermodynamic conjugate variables (e.g., pressure, chemical potential).}
	\item $u$: The scale-invariant internal energy of the system.
	\item $U$: The scale-dependent internal energy of the system.
	\item $Z$: The extended coordinate representing the global absolute scale of the system.
	\item $r_i,r,\rho$: The characteristic scaling weights associated with the variables $Q^i,U,Z$.
	\item $Z_0$: A strictly positive constant ($Z_0 > 0$) defining the gauge-fixing section.
	\item $\tau$: An arbitrary, dimensionless affine parameter.
	\item $t$: Real symplectization parameter.
\end{itemize}

\paragraph{Functions and Generators}
\begin{itemize}
	\item $H$: The homogeneous Hamiltonian function defined on the symplectic manifold $M$.
	\item $h$: The contact Hamiltonian function defined on the contact manifold $C$.
	\item $G$: The Noether charge associated with the continuous dilation (scaling) symmetry.
\end{itemize}

\paragraph{Vector Fields and Flows}
\begin{itemize}
	\item $X_r$: The canonical conformal vector field generating the gauge orbits (continuous dilations) in $M$.
	\item $X_G$: The Hamiltonian vector field associated with the Noether charge $G$, generating volume-preserving geometric scaling flows.
	\item $X_H$: The Hamiltonian vector field of $H$, generating the extended thermodynamic dynamics in $M$.
	\item $X_h$: The contact vector field of $h$, generating the physical, quasi-static thermodynamic processes on $C$.
\end{itemize}

}

\section{Constrained Dynamics~Formalism}
\label{sec:Constrained-Dynamics-Formalism}

Let us recall the Hamiltonian approach developed in~\cite{baldiotti2016}. Let us focus on a simple thermodynamic system whose equations of state are given in the energy representation: 
\begin{equation}
T=T\left(S,V,N\right)~,\ P=P\left(S,V,N\right)~,\ \mu=\mu\left(S,V,N\right),\label{eq:equations-state}
\end{equation}
where the extensive parameters are entropy $\left(S\right)$, volume $\left(V\right)$, particle or mole number $\left(N\right)$, while the intensive parameters are temperature $\left(T\right)$, pressure $\left(P\right)$, and~chemical potential $\left(\mu\right)$. To~incorporate the scale invariance of the formalism, it is common to work with scale-invariant quantities, that is, intensive quantities that represent densities of extensive variables. We note that any extensive quantities can be used to define these densities; for~example, the~scale-invariant quantities can be given per volume, per unit entropy, etc. The~most common choice is to use per-particle (or per-mole) ratios,
\begin{equation}
s=\frac{S}{N}~,\ v=\frac{V}{N}.\label{N}
\end{equation}

We first establish a dictionary between thermodynamical variables and coordinates $\left(q,p\right)$ in phase-space, where coordinates $q$ and momenta $p$ are:
\begin{equation}
q^{1}=s\,,\,q^{2}=v,~p_{1}=T\,,\,p_{2}=-P.\label{eq:thermodic}
\end{equation}
In 
 this way one defines the tautological one-form as $\theta=\sum_{i=1}^{2}p_{i}dq^{i}$ and the canonical symplectic form as $\omega=-d\theta$, which in local coordinates defines the Poisson brackets between functions $f$ and $g$ on phase space:
\begin{equation}
\left\{ f,g\right\} =\sum_{i=1}^{n}\frac{\partial f}{\partial q^{i}}\frac{\partial g}{\partial p_{i}}-\frac{\partial g}{\partial q^{i}}\frac{\partial f}{\partial p_{i}}\,,\,\,n=2,
\end{equation}
which can be trivially extended to $n>2$ to describe a more complex system (e.g., charged or multi component). The~equations of state \eqref{eq:equations-state} are therefore realized as primary constraints
\begin{equation}
p_{i}=f_{i}\left(q\right)\Leftrightarrow\phi_{i}\left(q,p\right)=p_{i}-f_{i}\left(q\right)\,,\,i=1,...,n.\label{eq:state-constraints}
\end{equation}
On the constraint surface $\phi_{i}=0$, the~tautological form $\theta$ is the differential of the internal energy $du$:
\begin{equation}
\left.\theta\right\vert _{\phi=0}=\sum_{i=1}^{2}p_{i}\left(q\right)dq^{i}=Tds-Pdv\equiv du.
\end{equation}

Taking into account the thermodynamic definition of intensive parameters in the chosen representation, $p_{i}=\frac{\partial u}{\partial q^{i}}\left(q\right)$, it follows that the tautological form is closed on the constraint~surface:
\begin{equation}
\left.d\theta\right\vert _{\phi=0}=\sum_{i,j=1}^{2}\frac{\partial p_{i}}{\partial q^{j}}\left(q\right)dq^{j}\wedge dq^{i}=\sum_{i,j=1}^{2}\frac{\partial^{2}u}{\partial q^{i}\partial q^{j}}dq^{i}\wedge dq^{j}\equiv0.
\end{equation}

In this symplectic setting, the~Legendre transformations between thermodynamic potentials become canonical transformations. For~instance, the~canonical transformation $\left(q^{1},p_{1}\right)\mapsto\left(q^{\prime},p^{\prime}\right)=\left(-p_{1},q^{1}\right)$ gives the tautological form
\begin{equation}
\theta^{\prime}=p^{\prime}dq^{\prime}+p_{2}q^{2}.
\end{equation}

The difference $\theta-\theta^{\prime}=d\left(q^{1}p_{1}\right)$ is a closed form, as~expected. On~the constraint surface one~recognizes
\begin{equation}
\left.\theta^{\prime}\right\vert _{\phi=0}=d\left(u-Ts\right),
\end{equation}
as the differential of the scale-invariant Helmholtz potential $f\left(T,v\right)=u-Ts$. Since two equilibrium states (with identical composition) of any thermodynamic system can be connected by some possible mechanical process, it follows that any curve connecting these states is a valid solution to the equations of motion of the mechanical analog. In~other words, the~action has an extremum for any curve connecting equilibrium states, so there are no physical degrees of freedom. As~we showed in~\cite{baldiotti2016}, this entails that either all $n$ constraints are first-class, or~there are $r$ first-class constraints and $k$ second-class constraints such that $n=r+k/2$. In~either case, the~Lagrangian is a total derivative, as~should be expected from the fact that action has an extremum for any trajectory connecting equilibrium states. Therefore, a~trajectory in configuration space corresponds to the idealized notion of quasi-static processes in equilibrium thermodynamics. As~a result of the absence of physical degrees of freedom, the~canonical Hamiltonian $H_{c}$ is proportional to constraints, \mbox{$H_{c}=\Sigma_{i=1}^{n}\lambda_{i}\phi_{i}$}, and~therefore, the~Lagrange function is first degree homogeneous in the velocities, $L\left(q,\lambda\dot{q}\right)=\lambda L\left(q,\dot{q}\right)$. Moreover, as~we have shown previously, $L=\frac{d\phi}{d{\tau}}$, where $p_{i}=\frac{\partial u}{\partial q^{i}}\left(q\right)$. Taking into account the thermodynamic definition of intensive parameters, $p_{i}=\frac{\partial u}{\partial q^{i}}\left(q\right)$, it follows that $L=\frac{du}{d{\tau}}$, that is, the~Lagrange function is the total derivative of the internal energy {with respect to the affine parameter $\tau$}. Thus, the~action functional is just the difference between the values of the internal energy of the initial and final states, which is not surprising in the light of the discussion of the previous~paragraph.

\section{Symplectic~Approach}
\label{sec:Symplectic-approach}

Let $M$ be a symplectic manifold of dimension $2n$ with symplectic form $\omega$. Consider a submanifold of dimension $n$ given by the zero level of the independent functions $\phi_{1},...,\phi_{n}:M\rightarrow\mathbb{R}$,
\begin{equation}
\mathcal{C}=\left\{ x\in M,\,\phi_{i}\left(x\right)=0\,,i=1,...,n\right\},
\end{equation}
which are involutive on $\mathcal{C}$, i.e.,
\begin{equation}
\left.\left\{ \phi_{i},\phi_{j}\right\} \right\vert _{\mathcal{C}}=0\,,\,\,\forall i,j.
\end{equation}

We will show that $\mathcal{C}$ is Lagrangian, that is, it is an $n$-dimensional isotropic submanifold of $M$ (see~\cite{abraham-marsden}). Since $\omega$ is non-degenerate, for~each $\phi_{i}$ there is a vector field $X_{\phi_{i}}\in TM$ defined~by
\begin{equation}
i_{X_{\phi_{i}}}\omega=d\phi_{i}\Leftrightarrow d\phi_{i}\left(Y\right)=\omega\left(X_{\phi_{i}},Y\right)\,\,\forall Y\in TM.
\end{equation}
In particular,
\begin{equation}
d\phi_{i}\left(X_{\phi_{j}}\right)=\omega\left(X_{\phi_{i}},X_{\phi_{j}}\right)=\left\{ \phi_{i},\phi_{j}\right\},
\end{equation}
are the Poisson brackets between the functions $\phi_{i}$ and $\phi_{j}$. Now consider a vector $v$ tangent to the constraint surface $\mathcal{C}$ at point $x$, $v\in T_{x}\mathcal{C}$. Then the directional derivative of $\phi_{i}$ in the direction of $v$ vanishes, $d\phi_{i}\left(v\right)=0$, since $\phi_{i}=0$ is constant in $\mathcal{C}$. That is,
\begin{equation}
\omega\left(X_{\phi_{i}},v\right)=0\,\,\forall v\in T_{x}M.
\end{equation}

We say that $X_{\phi_{i}}$ belongs to the symplectic complement of $T_{x}\mathcal{C}$, $X_{\phi_{i}}\in\left(T_{x}\mathcal{C}\right)^{\omega}$. In~fact, it is true that $\left(T_{x}\mathcal{C}\right)^{\omega}=\mathrm{span}\left\{ X_{\phi_{1}},...,X_{\phi_{n}}\right\}$.

\begin{lem}
If $\mathcal{C}$ is Lagrangian, then the tangent space at each point is equal to its symplectic complement, $T_{x}\mathcal{C}=\left(T_{x}\mathcal{C}\right)^{\omega}$, $x\in\mathcal{C}$.
\label{lem:lemma1} 
\end{lem}

\begin{proof}
First, we show that $T_{x}\mathcal{C}\subset\left(T_{x}\mathcal{C}\right)^{\omega}$ (as a vector space). Since $\mathcal{C}$ is Lagrangian, given $X\in T_{x}\mathcal{C}$, $\left.\omega\right\vert _{\mathcal{C}}\left(X,Y\right)=0$ for all $Y\in T_{x}\mathcal{C}$, that is, $X\in\left(T_{x}\mathcal{C}\right)^{\omega}$. On~the other hand, since $\omega$ is symplectic (non-degenerate), one has
\begin{equation}
\dim\left(T_{x}\mathcal{C}\right)^{\omega}+\dim T_{x}\mathcal{C}=\dim T_{x}M=2n~.
\end{equation}
Since $\dim T_{x}\mathcal{C}=n$, it follows that $\dim\left(T_{x}\mathcal{C}\right)^{\omega}=n$. Then, $T_{x}\mathcal{C}=\left(T_{x}\mathcal{C}\right)^{\omega}$. 
\end{proof}

\begin{thm}
$\mathcal{C}$ is Lagrangian if and only if the constraints $\phi_{i}$ are involutive, $\left.\left\{ \phi_{i},\phi_{j}\right\} \right\vert _{\mathcal{C}}=0$.
\label{thm:-is-Lagrangian} 
\end{thm}

\begin{proof}
By Lemma \ref{lem:lemma1}, we have $T_{x}\mathcal{C}=\left(T_{x}\mathcal{C}\right)^{\omega}$. Therefore, the~fields $X_{\phi_{i}}$ are tangent to the constraint surface. Since $\mathcal{C}$ is Lagrangian, $\left.\omega\right\vert _{\mathcal{C}}=0$. Thus, $\left.\omega\right\vert _{\mathcal{C}}\left(X_{\phi_{i}},X_{\phi_{j}}\right)=0$ and it follows that $\left.\left\{ \phi_{i},\phi_{j}\right\} \right\vert _{\mathcal{C}}=0$. If~on the other hand the constraints are involutive in $\mathcal{C}$, $\left.\left\{ \phi_{i},\phi_{j}\right\} \right\vert _{\mathcal{C}}=0$, then it follows that the fields $X_{\phi_{i}}$ are tangent, because~they preserve the constraint surface,
\begin{equation}
X_{\phi_{i}}\left(\phi_{j}\right)=d\phi_{j}\left(X_{\phi_{i}}\right)=\left.\omega\right\vert _{\mathcal{C}}\left(X_{\phi_{j}},X_{\phi_{i}}\right)=\left.\left\{ \phi_{i},\phi_{j}\right\} \right\vert _{\mathcal{C}}=0.
\end{equation}
That is, $X_{\phi_{i}}\in T_{x}\mathcal{C}$. As~the $X_{\phi_{i}}$ are linearly independent at each point $x\in\mathcal{C}$, they constitute a basis for the tangent space $T_{x}\mathcal{C}$. Therefore, $\dim\mathcal{C}=n=\frac{1}{2}\dim M$ and $\left.\omega\right\vert _{\mathcal{C}}=0$, that is, $\mathcal{C}$~is Lagrangian. 
\end{proof}

\subsection*{Ideal and van der Waals~Gases}

In this subsection, we use the thermodynamic dictionary from Section~\ref{sec:Constrained-Dynamics-Formalism} given in (\ref{eq:thermodic}). The~constraint surface for the ideal gas is described by the constraints:
\begin{equation}
\phi_{1}=p_{2}+Ae^{\frac{2}{3}q^{1}}\left(q^{2}\right)^{-\frac{5}{3}}=0\,,\,\,\phi_{2}=p_{1}-Ae^{\frac{2}{3}q^{1}}\left(q^{2}\right)^{-\frac{2}{3}}=0,
\end{equation}
{where $A$ is an integration constant (see~\cite{baldiotti2016}). They are given} on the symplectic manifold $M=\mathbb{R}^{4}$ given in a local Darboux chart $\left(q^{i},p_{i}\right)$ and symplectic form { $\omega=\sum_{i=1}^{2} dq^{i}\wedge dp_{i} $} for which $\left\{ \phi_{1},\phi_{2}\right\} =0.$ That is, the~fields
\begin{align}
X_{\phi_{1}} & =-\frac{\partial}{\partial q^{2}}-\frac{5}{3}Ae^{\frac{2}{3}q^{1}}\left(q^{2}\right)^{-\frac{8}{3}}\frac{\partial}{\partial p_{2}}+\frac{2}{3}Ae^{\frac{2}{3}q^{1}}q^{-\frac{5}{3}}\frac{\partial}{\partial p_{1}}~,\\
X_{\phi_{2}} & =-\frac{\partial}{\partial q^{1}}-\frac{2}{3}Ae^{\frac{2}{3}q^{1}}\left(q^{2}\right)^{-\frac{2}{3}}\frac{\partial}{\partial p_{1}}+\frac{2}{3}Ae^{\frac{2}{3}q^{1}}\left(q^{2}\right)^{-\frac{5}{3}}\frac{\partial}{\partial p_{2}},
\end{align}
are involutive $\omega\left(X_{\phi_{1}},X_{\phi_{2}}\right)\equiv0$. Since the constraints are first-class (involutive), the~constraint surface given by $\phi_{1}=\phi_{2}=0$ is a Lagrange submanifold (Theorem \ref{thm:-is-Lagrangian}). Note that the isotropy condition $\omega=0$ on the constraint surface is equivalent to the Maxwell relation
\begin{equation}
-\frac{\partial\phi_{2}}{\partial q^{2}}+\frac{\partial\phi_{1}}{\partial q^{1}}=0\Leftrightarrow-\frac{\partial p_{1}\left(q\right)}{\partial q^{2}}+\frac{\partial p_{2}\left(q\right)}{\partial q^{1}}\Leftrightarrow  {\frac{\partial T}{\partial v}=-\frac{\partial P}{\partial s}}.
\end{equation}

Thus, the~Lagrange submanifold which describes the ideal gas is
\begin{equation}
L_{IG}=\left\{ \left(q,p\right)\in\mathbb{R}^{4},\,\,\phi_{1}=\phi_{2}=0\right\}.
\end{equation}

In $L_{IG}$, the~canonical form $\theta=\sum_{i=1}^{2}p_{i}dq^{i}$ is
\begin{equation}
\left.\theta\right\vert _{L_{IG}}=d\left(\frac{3}{2}Ae^{\frac{2}{3}q^{1}}\left(q^{2}\right)^{-\frac{2}{3}}\right)=du,
\end{equation}
where $u=\frac{3}{2}Ae^{\frac{2}{3}q^{1}}\left(q^{2}\right)^{-\frac{2}{3}}$ is the internal energy of the ideal gas. Consider the transformation $\psi:\mathbb{R}^{4}\rightarrow\mathbb{R}^{4}$, $\psi\left(q^{\prime},p^{\prime}\right)=\left(q,p\right)$, given by
\begin{equation}
q^{2}=q^{\prime2}-b\,,\,\,p_{2}=p_{2}^{\prime}-\left(q^{\prime2}\right)^{-2}\,,\,\,p_{1}=p_{1}^{\prime}\,,\,\,q^{1}=q^{\prime1},
\end{equation}
which is a symplectomorphism $\omega^{\prime}=\psi^{\ast}\omega$. The~pullback of the constraints
\begin{eqnarray}
\phi_{1}^{\prime} & = & \psi^{\ast}\phi_{1}=p_{2}^{\prime}-\left(q^{\prime2}\right)^{-2}+Ae^{\frac{2}{3}q^{\prime1}}\left(q^{\prime2}-b\right)^{-5/3},\nonumber \\
\phi_{2}^{\prime} & = & \psi^{\ast}\phi_{2}=p_{1}^{\prime}-Ae^{\frac{2}{3}q^{\prime1}}\left(q^{\prime2}-b\right)^{-\frac{2}{3}},
\end{eqnarray}
define the Lagrange submanifold
\begin{equation}
L_{WG}=\left\{ \left(q^{\prime},p^{\prime}\right)\in\mathbb{R}^{4},\,\,\phi_{1}^{\prime}=\phi_{2}^{\prime}=0\right\}.
\end{equation}

In $L_{WG}$,
\begin{equation}
\left.\theta^{\prime}\right\vert _{{L_{WG}}}=d\left(\frac{3}{2}Ae^{\frac{2}{3}q^{\prime1}}\left(q^{\prime2}-b\right)^{-\frac{2}{3}}\right)=du^{\prime},
\end{equation}
where $u^{\prime}=\frac{3}{2}Ae^{\frac{2}{3}q^{\prime1}}\left(q^{\prime2}-b\right)^{-\frac{2}{3}}$ is the internal energy of the van der Waals gas. Let $TL_{IG}$ be the tangent manifold to $L_{IG}$, then a local basis of fields is given by
\begin{equation}
\left\{ e_{1}=\frac{\partial}{\partial q^{1}}+\frac{\partial p_{2}}{\partial q^{1}}\frac{\partial}{\partial p_{2}}+\frac{\partial p_{1}}{\partial q^{1}}\frac{\partial}{\partial p_{1}},\,e_{2}=\frac{\partial}{\partial q^{2}}+\frac{\partial p_{2}}{\partial q^{2}}\frac{\partial}{\partial p_{2}}+\frac{\partial p_{1}}{\partial q^{2}}\frac{\partial}{\partial p_{1}}\right\}.
\end{equation}

Analogously, a~local basis of fields for $TL_{WG}$ is
\begin{equation}
\left\{ e_{1}^{\prime}=\frac{\partial}{\partial q^{\prime1}}+\frac{\partial p_{1}^{\prime}}{\partial q^{\prime1}}\frac{\partial}{\partial p_{2}^{\prime}}+\frac{\partial p_{1}^{\prime}}{\partial q^{\prime1}}\frac{\partial}{\partial p_{1}^{\prime}},\,e_{2}^{\prime}=\frac{\partial}{\partial q^{\prime2}}+\frac{\partial p_{2}^{\prime}}{\partial q^{\prime2}}\frac{\partial}{\partial p_{2}^{\prime}}+\frac{\partial p_{1}^{\prime}}{\partial q^{\prime2}}\frac{\partial}{\partial p_{1}^{\prime}}\right\}.
\end{equation}

Let us call $\phi:L_{WG}\rightarrow L_{IG}$ the restriction of $\psi$ to the submanifold {$L_{WG}$}. Explicitly,
\begin{equation}
\phi\left(q^{\prime},p^{\prime}\left(q^{\prime}\right)\right)=\left(q\left(q^{\prime}\right),p\left(q^{\prime}\right)\right),
\end{equation}
where
\begin{equation}
p_{1}\left(q^{\prime}\right)=Ae^{\frac{2}{3}q^{\prime1}}\left(q^{\prime2}-b\right)^{-\frac{2}{3}}\,,\,\,p_{2}\left(q^{\prime}\right)=-Ae^{\frac{2}{3}q^{\prime1}}\left(q^{\prime2}-b\right)^{-5/3}.
\end{equation}

The tangent map $d\phi$ at the basis vectors of $TL_{WG}$ maps $d\phi\left(e_{i}^{\prime}\right)=e_{i}$, $i=1,2$. Therefore $\left[d\phi\right]=I_{2}$ is the identity map and has constant rank. This shows that the two Lagrangian submanifolds are diffeomorphic, which can be seen as a consequence of Weinstein's theorem~\cite{abraham-marsden,weinstein,cannas}.

While the result $[d\phi] = I_2$ might appear trivial from a purely mathematical standpoint, being a natural consequence of Weinstein's theorem for nearby Lagrangian submanifolds, the~explicit construction of $\phi$ carries physical content. Mathematically, Weinstein's theorem guarantees only the local existence  of such a mapping. Physically, however, obtaining the explicit symplectomorphism provides the exact generating function that encodes macroscopic intermolecular interactions. The~fact that the tangent map reduces to the identity matrix perfectly captures the asymptotic physical limit of the system: at large volumes, the~non-linear interaction terms vanish, and~the van der Waals gas smoothly converges to the ideal gas. Furthermore, providing the explicit non-linear form of $\phi$ allows us to probe the global geometric properties of the state space, including~folds in the Lagrangian submanifold that physically correspond to phase transitions and critical phenomena.

The explicit construction of this mapping naturally raises the question of its global validity: under precisely which geometric and physical criteria can two distinct macroscopic thermodynamic systems be related by a global symplectomorphism? Because the Lagrangian submanifold is globally defined as the graph of the differential of the internal energy over the extensive variables, its intrinsic geometry remains smooth. Therefore, phase transitions, which manifest as caustic singularities only upon projection via the Legendre transform (e.g., the~spinodal curve), do not inherently obstruct the existence of a global smooth map. We note that, in the above example, there is no global topological obstruction, that is, the~physical domains of the extensive variables must be diffeomorphic. While the shifted volume domain of the van der Waals gas is diffeomorphic to the volume domain of the ideal gas, a~global map would fail if one attempted to relate an unbounded gas to a system with a compact state space, such as a paramagnet.

 In general, one has

\begin{thm}
Let $\phi:M\rightarrow M^{\prime}$ be a symplectomorphism, i.e.,~$\phi^{*}\omega^{\prime}=\omega$. If~$L^{\prime}$ is a Lagrange submanifold of $M^{\prime}$, then $\phi^{-1}\left(L\right)$ is a Lagrange submanifold of $M$.
\end{thm}

\begin{proof}
Since $\phi$ is a symplectomorphism, $\dim M=\dim M^{\prime}=2n$. The~restriction of $\phi$ to $L$~is a diffeomorphism, then $\dim L=\dim L^{\prime}$. Since $L^{\prime}$ is Lagrangian, $\dim L=\dim L^{\prime}=n$. Let us show that $L$ is isotropic. Let $i^{\prime}:L^{\prime}\hookrightarrow M^{\prime}$ be the inclusion, then $i^{\prime\ast}\omega^{\prime}=0$. And~let ${\varphi}:L\rightarrow L^{\prime}$ be the diffeomorphism given by $\phi\circ i=i^{\prime}\circ{\varphi}$, in~which $i:L\hookrightarrow M$. Thus,
\begin{equation}
i^{\ast}\omega=i^{\ast}\phi^{\ast}\omega^{\prime}=\left(\phi\circ i\right)^{\ast}\omega^{\prime}=\left(i^{\prime}\circ{\varphi}\right)^{\ast}\omega^{\prime}={\varphi}^{\ast}i^{\prime\ast}\omega^{\prime}=0.
\end{equation}
Since ${\varphi}^{\ast}$ is an isomorphism, it follows that $i^{\ast}\omega=0$. 
\end{proof}

\section{Contactization}
\label{sec:Contactization}

The symplectic formalism can be naturally extended to the natural contact formalism of thermodynamics~\cite{brav1} by the process of contactization. Given a symplectic manifold $\left(M,\omega\right)$ of dimension $2n$, the~contactization of $\left(M,\omega\right)$ is the process of constructing the contact manifold $\left(C,\alpha\right)$ of dimension $2n+1$ where $M$ is the base manifold of a principal fiber bundle defined by $C$ such that $d\alpha=\pi^{*}\omega$, and~whose fibers are generated by the Reeb field $R$, $i_{R}\alpha=1$ \cite{cannas}. This can be easily understood in the case of the trivial or product bundle. Let $Q$ be a configuration manifold of dimension $n$ (for instance, the~configuration manifold of a thermodynamic system, locally given by its extensive variables). And~let $M=T^{\ast}Q$ be the cotangent fiber bundle, and~suppose that $\omega$ is exact, that is, $\omega=-d\theta$, where $\theta$ is the canonical (or tautological) $1$-form. The~contactization is the product $C=M\times\mathbb{R}$ and the contact form $\alpha\in\Omega^{1}\left(C\right)$ is globally defined by $\alpha=dz-\pi^{\ast}\theta$ where $z$ is a coordinate in $\mathbb{R}$ and $\pi:C\rightarrow M$ is the natural projection. Note that
\begin{equation}
d\alpha=-d\pi^{\ast}\theta=\pi^{\ast}\omega.
\end{equation}

In local coordinates, $\left(q^{i},p_{i},z\right)$, the~contact form is {$\alpha=dz-\sum_{i=1}^{n}p_{i}dq^{i}$} and {$d\alpha=-\sum_{i=1}^{n}dp_{i}\wedge dq^{i}=\omega.$} The projection is given by $\left(q^{i},p_{i},z\right)\mapsto\left(q^{i},p_{i}\right)$. There is a natural connection between Lagrangian submanifolds $L$ of the symplectic manifold $M$ and the Legendre submanifolds $\Lambda$ of the contact bundle $C$. Let $\Lambda$ be a Legendre submanifold, that is, $\dim\Lambda=n$ and $j^{*}\alpha=0$, where $j:\Lambda\hookrightarrow C$ is the inclusion~map.

\begin{thm}
Let $L=\pi\left(\Lambda\right)$. Then $L$ is an immersed submanifold in $M$ of dimension $n$, and~${\varphi}:\Lambda\rightarrow L$ is the diffeomorphism given locally by
\begin{equation}
\pi\circ j=i\circ{\varphi}
\end{equation}
where $i:L\hookrightarrow M$.
\end{thm}

\begin{proof}
In order to show that $L$ is an immersed manifold of dimension $n$, we need to show that $\Phi=\pi\circ j$ is an immersion, that is, we need to show that $T_{\lambda}\Phi:T_{\lambda}\Lambda\rightarrow T_{\Phi\left(\lambda\right)}M$ is injective. Let $v\in T_{\lambda}\Lambda$ such that $T_{\lambda}\Phi\left(v\right)=0$. By~the chain rule,
\begin{equation}
T_{\lambda}\left(\pi\circ j\right)\left(v\right)=T_{j\left(\lambda\right)}\pi\circ T_{\lambda}j\left(v\right)=0.
\end{equation}
By calling $w=T_{\lambda}j\left(v\right)$, this shows that $w$ is in the kernel of $T_{j\left(\lambda\right)}\pi$. Therefore, $w$ is vertical, that is, it is proportional to the Reeb vector $R$. Since $\Lambda$ is a Legendre submanifold, $j^{\ast}\alpha=0,$ and therefore $\alpha\left(T_{\lambda}j\cdot v\right)=0$ for all $v\in T_{\lambda}\Lambda$. In~particular, $\alpha\left(w\right)=0$. Since $w=kr,$ for $k\in\mathbb{R}$, and~$\alpha\left(R\right)=1,$ it follows that $k=0$. Since $j$ is an immersion, $T_{\lambda}j$ is injective, and~therefore $v=0$. This shows that $\Phi$ is an immersion by the Local Immersion Theorem~\cite{abraham-marsden}. Thus, $L=\Phi\left(\Lambda\right)=\pi\left(j\left(\Lambda\right)\right)=\pi\left(\Lambda\right)$ is an immersion and $\dim L=n$. In~order to show that ${\varphi}:\Lambda\rightarrow L$ is a local diffeomorphism, we need to show that for $\lambda\in\Lambda$, $T_{\lambda}{\varphi}:T_{\lambda}\Lambda\rightarrow T_{{\varphi}\left(\lambda\right)}L$ is an isomorphism. Since $\dim\Lambda=\dim L=n$, it is enough to show that $T_{\lambda}{\varphi}$ is injective. Let $v$ be in the kernel of $T_{\lambda}{\varphi}$, that is, $T_{\lambda}\rho\left(v\right)=0$. By~the chain rule,
\begin{equation}
T_{\lambda}\left(\pi\circ j\right)=T_{\lambda}\left(i\circ{\varphi}\right)=T_{{\varphi}\left(\lambda\right)}i\circ T_{\lambda}{\varphi}.
\end{equation}
Applying to $v$, and~using the fact that we have previously established that $T_{\lambda}\left(\pi\circ j\right)$ is injective, it follows that $T_{\lambda}{\varphi}$ is injective. Therefore, by~the inverse function theorem, ${\varphi}$ is a local diffeomorphism. 
\end{proof}

\begin{prop}
The manifold $L$ is Lagrangian.
\label{prop:lagrangian}
\end{prop}

\begin{proof}
It is enough to show that $i^{\ast}\omega=0$ (isotropy), for~we have already shown that $L$ has the correct dimension, $\dim L=n$. Since $L$ is defined in terms of a projection, let us examine the Legendre submanifold $\Lambda$, where $j^{\ast}\alpha=0$. Taking the exterior derivative,
\begin{equation}
dj^{\ast}\alpha=j^{\ast}d\alpha=j^{\ast}\pi^{\ast}\omega=\left(\pi\circ j\right)^{\ast}\omega=0.
\end{equation}
Because $\pi\circ j=i\circ{\varphi}$,
\begin{equation}
\left(i\circ{\varphi}\right)^{\ast}\omega={\varphi}^{\ast}i^{\ast}\omega=0.
\end{equation}
Since ${\varphi}$ is a local diffeomorphism, $\rho^{\ast}$ is a linear isomorphism on differential forms ${\varphi}^{\ast}:\Lambda^{2}\left(T_{{\varphi}\left(\lambda\right)}L\right)\rightarrow\Lambda^{2}\left(T_{\lambda}\Lambda\right)$. Thus, $i^{\ast}\omega=0$. 
\end{proof}

\subsection*{Ideal and van der Waals~Gases}

Let us contactize the symplectic structure which describes the ideal gas. Consider $C=M\times\mathbb{R}\simeq\mathbb{R}^{5}$, in~which we added the coordinate $u$ (internal energy), and~the contact $1$-form $\alpha=du-\sum_{i=1}^{2}p_{i}dq^{i}$. The~Legendre submanifold $\Lambda$ is a submanifold of $\mathbb{R}^{5}$ of dimension $2$ which is given by $j^{\ast}\alpha=0$. In~order to obtain a representation in scale-invariant coordinates, we need to choose a generating function, like
\begin{equation}
u=f\left(q\right)=\frac{3}{2}A\left(q^{2}\right)^{-\frac{2}{3}}e^{\frac{2}{3}q^{1}}.
\end{equation}

In this case $j^{\ast}\alpha=0$ becomes
\begin{equation}
du=\sum_{i=1}^{2}p_{i}\left(q\right)dq^{i}=\sum_{i=1}^{2}\frac{\partial u}{\partial q^{i}}dq^{i},
\end{equation}
that is,
\begin{equation}
p_{2}\left(q\right)=-A\left(q^{2}\right)^{-\frac{5}{3}}e^{\frac{2}{3}q^{1}}\,,\,\,p_{1}\left(q\right)=A\left(q^{2}\right)^{-\frac{2}{3}}e^{\frac{2}{3}q^{1}}
\end{equation}
and the Legendre submanifold is given in these coordinates as
\begin{equation}
\Lambda=\left\{ \left(u,q,p\right)\in\mathbb{R}^{5},u=f\left(q\right),\,p_{i}=\frac{\partial f}{\partial q^{i}}\,,i=1,2\right\}.
\end{equation}

Since the function $u$ is smooth and globally defined, we can identify the Legendre submanifold with its projection on $\mathbb{R}^{3}$, $\left(u,q,p\right)\mapsto\left(q,u\right)$. This projection is the graph of the function $u\left(q\right)$.

The Lagrange submanifold $L$ can be obtained by the projection $\pi:C\rightarrow M$, $\pi\left(u,q,p\right)=\left(q,p\right)$ (Proposition \ref{prop:lagrangian}), which satisfies $d\alpha=\pi^{\ast}\omega$,
\begin{equation}
L=\pi\left(\Lambda\right)=\left\{ \left(q,p\right)\in M,\exists u\in\mathbb{R}|\left(u,q,p\right)\in L\right\}.
\end{equation}

In coordinates, $\omega=\sum_{i=1}^{2}dq^{i}\wedge dp_{i}$ and
\begin{align}
\left.\omega\right\vert _{L} & =\sum_{i=1}^{2}dq^{i}\wedge dp_{i}\left(q\right)\nonumber \\
 & =\left(\frac{\partial p_{2}}{\partial q^{1}}-\frac{\partial p_{1}}{\partial q^{2}}\right)dq^{2}\wedge dq^{1}\equiv0.
\end{align}

Therefore, we can understand the Lagrange submanifold as the description of the thermodynamic system at a certain height $u$. The~contactization lifts the Lagrange submanifold to the Legendre submanifold; it is the thermodynamic description which contains the fundamental relation. In~the case of the ideal gas, since the Legendre submanifold $\Lambda$ is globally generated by a function (the internal energy $u\left(q\right)$), both $\Lambda$ and $L$ can be globally parametrized by the base coordinates $q$. Let $Q$ be the configuration space in coordinates $q$. A~parametrization for $\Lambda$ is $i_{\Lambda}:Q\rightarrow C$ given by
\begin{equation}
i_{\Lambda}\left(q\right)=\left(q,p,u\right)=\left({q^1,q^2},p_{1}=\frac{2}{3}u,p_{2}=-\frac{2}{3}\frac{u}{q^{2}},u\right)
\end{equation}
in which $u=u\left(q\right)$ is the internal energy function. A~parametrization for $L$ is $i_{L}:Q\rightarrow M$ given by 
\begin{equation}
i_{L}\left(q\right)=\left(q,p\right)=\left({q^1,q^2},p_{1}=\frac{2}{3}u,p_{2}=-\frac{2}{3}\frac{u}{q^{2}}\right).
\end{equation}

Both define regular parametrizations (of maximum rank). The~diffeomorphism ${\varphi}$ is global and it is simply given by
\begin{equation}
{\varphi}\left(q,p,u\right)=\left(q,p\right)\,,\,\,{\varphi}^{-1}\left(q,p\right)=\left({q^1,q^2},p,u=\frac{3}{2}p_{1}\right).
\end{equation}

\section{Symplectization}
\label{sec:Symplectization}

In this section, we provide another perspective in which symplectization is understood as a formal procedure for providing a thermodynamic system given in terms of scale-invariant quantities with a description in terms of homogeneous quantities of arbitrary degree. Consider a contact manifold $\left(C,\alpha\right)$ of dimension $2n+1$, with~contact form $\alpha$ given in local Darboux coordinates as
\begin{equation}
\alpha=du-\sum_{i=1}^{n}p_{i}dq^{i}.\label{eq:contact-form}
\end{equation}

Let $\Lambda$ be the Legendre submanifold given by the set of equations
\begin{equation}
u=u\left(q\right),\,\,p_{i}=\frac{\partial u}{\partial q^{i}}\,,\,\,{u\left(\lambda^{r_q} q\right)=u\left(q\right)},
\end{equation}
that is $u\left(q\right)$ is homogeneous of degree $0$. Let $\left(M,\omega\right)$ be a symplectic manifold given as \mbox{$M={C}\times\mathbb{R}$}, where the symplectic form is $\omega=d\left(e^{rt}\alpha\right)$. Now consider the symplectomorphism $\psi:\left(t,u,q^{i},p_{i}\right)\mapsto\left(Z,\mu,Q^{i},P_{i}\right)$,
\begin{equation}
Z=e^{\rho t}\,,\,\,\mu=\rho^{-1}e^{\left(r-\rho\right)t}\left(ru-\sum_{i=1}^{n}r_{i}p_{i}q^{i}\right)\,,\,\,\,\,Q^{i}=e^{r_{i}t}q^{i}\,,\,\,P_{i}=e^{\left(r-r_{i}\right)t}p_{i},\label{eq:transf-can}
\end{equation}
such that $\omega=\psi^{\ast}\omega^{\prime}$. The~symplectic form in the new coordinates is canonical
\begin{equation}
\omega^{\prime}=dZ\wedge d\mu+\sum_{i=1}^ndQ^{i}\wedge dP_{i}.
\end{equation}

Here $Z=e^{\rho t}$ is a scale thermodynamic quantity analogous to $N$ in (\ref{N}), and~$\mu$ is therefore the analog of the chemical potential. In~terms of the canonical one forms $\omega=-d\theta$ and $\omega^{\prime}=-d\theta^{\prime}$, one finds $\psi^{\ast}\theta^{\prime}-\theta=dU$, where $U=e^{rt}u$. In~terms of the new variables, the~generating function $U$ satisfies
\begin{equation}
rU=\rho\mu Z+\sum_{i=1}^nr_{i}P_{i}Q^{i}.\label{eq:euler-relation}
\end{equation}

In the Lagrangian submanifold, the~generating function $U$ satisfies the Euler relation. In~fact, the~new internal energy $U\left(Q,Z\right)$ is homogeneous of degree $r$, when the coordinates $Q^{i}$ are scaled with weight $r_{i}$, while $Z$ is scaled with degree $\rho$, since
\begin{equation}
U\left(Q^{i},Z\right)=Z^{r/\rho}u\left(\frac{Q^{i}}{Z^{r_{i}/\rho}}\right),\label{energy}
\end{equation}
implies
\begin{equation}
U\left(\lambda^{r_{i}}Q^{i},\lambda^{\rho}Z\right)=\lambda^{r}U\left(Q,Z\right).
\end{equation}

\subsection{Gauge~Fixing}

The symmetry generator for the dilations
\begin{equation}
Q^{i}\mapsto\lambda^{r_{i}}Q^{i}\,,\,\,Z\mapsto\lambda^{\rho}Z\,,\,\,P_{i}\mapsto\lambda^{r-r_{i}}P_{i}\,,\,\,\mu\mapsto\lambda^{r-\rho}\mu\label{eq:dilations}
\end{equation}
is the ``canonical conformal field''
\begin{equation}
X_{r}=\sum_{i=1}^nr_{i}Q^{i}\frac{\partial}{\partial Q^{i}}+\rho Z\frac{\partial}{\partial Z}+\sum_{i=1}^n\left(r-r_{i}\right)P_{i}\frac{\partial}{\partial P_{i}}+\left(r-\rho\right)\mu\frac{\partial}{\partial\mu},\label{eq:liouville}
\end{equation}
since $\mathcal{L}_{X_{r}}\theta^{\prime}=r\theta$, or~equivalently, $\mathcal{L}_{X_{r}}\omega^{\prime}=r\omega^{\prime}$. The~hypersurfaces $\Sigma_{0}$ of constant $Z_{0}>0$, given by the constraints $\Phi=Z-Z_{0}=0$, are transverse to the flux of $X_{r}$, i.e.,~\mbox{$\left.X_{r}\left(\Phi\right)\right|_{\Sigma_{0}}=\rho Z_{0}$}, such that each orbit of $X_{r}$ intercepts the gauge fixing surface $\Sigma_{0}$ only once. Thus, the~coordinate $Z=e^{\rho t}$ acts as a global scale of the system, and~by gauge fixing $Z=Z_{0}$ one is taking the {quotient} of the symplectic manifold by orbits of $X_{r}$ to obtain a description in terms of scale-invariant quantities, which is effectively the initial contact description. This can be seen explicitly by noting that the pullback of the one-form $i_{X_{r}}\omega^{\prime}$ to $\Sigma_{0}$ is $\left.\alpha\right|_{\Sigma_{0}}=rZ_{0}^{r/\rho}\alpha$. One can also regard the symplectic manifold $M$ as a fiber bundle over the contact manifold $C$, with~the structure group being the group of dilations $\mathbb{R}_{>0}$, which acts on $M$ through the flow of $X_{r}$. 

{To formalize this gauge-theoretic interpretation, it is crucial to clarify the nature of the physical redundancy and the topology of the resulting quotient space. Physically, the~redundant degree of freedom is the absolute global scale (i.e., the~total size) of the system. Because~equilibrium thermodynamic potentials obey strict Euler homogeneity, scaling the extensive dimensions of the system does not alter its intrinsic equations of state. Consequently, any two points lying on the same one-dimensional orbit generated by the conformal field $X_r$ describe the exact same thermodynamic state. This physical equivalence justifies interpreting the scaling orbits as gauge-equivalent states. Furthermore, the~mathematical process of taking the quotient $M/\mathbb{R}_{>0}$ presents no global topological obstructions. Because~the structure group of dilations, $\mathbb{R}_{>0}$, is contractible, the~principal fiber bundle over the contact manifold $C$ is topologically trivial. Therefore, the~gauge-fixing condition $Z = Z_0 > 0$ defines a smooth, globally well-defined section that embeds $C$ back into $M$ without singularities or global obstructions.}

Let the continuous scale parameter be $\lambda=e^{s}>0$. The~action of this dilation group on $M$ is given by the dilations (\ref{eq:dilations}), which scale the 
coordinates. Under~this action, the~orbits of $X_{r}$ correspond exactly to the one-dimensional curves generated by varying $\lambda$. The~quotient space $M/\mathbb{R}_{>0}$ is diffeomorphic to $C$, since the smooth projection $\pi:M\rightarrow C$ is well defined, given that the functions $u\left(Q,P,Z,\mu\right)$, $q^{i}\left(Q,Z\right)$ and $p_{i}\left(P,Z\right)$ are invariant under dilations. Consequently, the~projection map $\pi$ collapses each dilation orbit in $M$ into a single, unique thermodynamic state in $C$. Within~this fiber bundle framework, the~gauge fixing condition $Z=Z_{0}$ acts as the choice of a global section, embedding $C$ back into $M$. Furthermore, the~Noether charge~\cite{abraham-marsden,souriau} is given by $G=i_{X_{r}}\theta^{\prime}$, which evaluates to
\begin{equation}
G=\rho\mu Z+\sum_{i=1}^{n}r_{i}P_{i}Q^{i}.
\end{equation}

When restricted to the Lagrange submanifold $L$, this charge coincides with the scaled energy function, $\left.G\right\vert _{L}=rU\left(Q,Z\right)$. The~Hamiltonian vector field associated to this function~is
\begin{equation}
X_{G}=\sum_{i=1}^nr_{i}\left(Q^{i}\frac{\partial}{\partial Q^{i}}-P_{i}\frac{\partial}{\partial P_{i}}\right)+\rho\left(Z\frac{\partial}{\partial Z}-\mu\frac{\partial}{\partial\mu}\right),
\end{equation}
since $i_{X_{G}}\omega^{\prime}=-dG$ and $\mathcal{L}_{X_{G}}\omega^{\prime}=0$. Unlike the conformal field $X_{r}$, the~Hamiltonian vector field $X_{G}$ strictly preserves the symplectic volume. Physically, its integral flow represents a canonical scale transformation that inversely scales extensive and intensive variables, maintaining the total energy invariant. It is precisely this property that allows us to treat scale-invariant quantities as invariants under dilations. In~fact, the~scale-invariant quantities $q^{i}=Z^{-r_{i}/\rho}Q^{i}$ and $p_{i}=Z^{\left(r_{i}-r\right)/\rho}P_{i}$ commute with $G$, $\left\{ q^{i},G\right\} =\left\{ p_{i},G\right\} =0$. To~describe thermodynamic processes that preserve the equilibrium surface while respecting the scaling symmetries, we introduce a Hamiltonian function $H$ which commutes with $G$, vanishes on $L$ and is homogeneous, $X_{r}\left(H\right)=rH$. A~natural choice for such a function focusing on the scaling variable $Z$ is
\begin{equation}
H=Z^{1-r/\rho}\rho\left(\mu\left.G\right|_{L}-\frac{\partial U}{\partial Z}G\right).
\label{H}\end{equation}

{
It is the simplest, strictly linear combination of the intrinsic thermodynamic constraints related to the scaling variables. The~commutator of $K=\lambda_{1}\left(\mu-\frac{\partial U}{\partial Z}\right)+\lambda_{2}\left(G-\left.G\right|_{L}\right)$ with G closes in the Poisson algebra iff $\lambda_{1}=\left.G\right|_{L}$ and $\lambda_{2}=-\frac{\partial U}{\partial Z}$. For~these values of $\lambda_{1}$ and $\lambda_{2}$, one has $K=\mu\left.G\right|_{L}-\frac{\partial U}{\partial Z}G$. For~the final Hamiltonian $H$ to commute with $G$, we multiply $K$ by an integrating factor $g\left(Z\right)$, $H=g\left(Z\right)K$. Then, $\left\{ H,G\right\} =0$ implies $g\left(Z\right)\sim Z^{1-r/\rho}$. The~final expression is $H=\rho Z^{1-r/\rho}K$ where $\rho$  is a normalization constant.}

{The expression for $H$ is not unique, given the starting ansatz. For~an alternative Hamiltonian of the form $\tilde{H} = f \cdot H$ to remain valid, the~modifying factor $f$ must satisfy two strict geometric conditions: it must be homogeneous of degree zero to preserve the extended scaling symmetry ($X_r(f) = 0$), and~it must strongly commute with the dilation generator ($\{f, G\} = X_G(f) = 0$). Consequently, $f$ must be annihilated by the difference of these two  differential operators, $(X_r - X_G) = \sum p_i \frac{\partial}{\partial p_i}$. By~Euler's theorem, this severely restricts $f$ to be a scale-invariant function constructed exclusively from non-linear products and ratios of variables whose net scaling degree exactly cancels out (for instance, a~dimensionless parameter like $q^i q^j / (q^k)^2$ if $r_i + r_j = 2r_k$). While mathematically allowed, placing such  specific fractional parameters into the  Hamiltonian would arbitrarily bias the dynamics. Thus, setting $f=1$ emerges as the  natural choice, preserving the unbiased physical geometry through the simplest linear combination of the thermodynamic scaling~constraints.}

Its Hamiltonian flow generates scale-invariant quasi-static processes along the equilibrium manifold.
Since by construction $H$ commutes with the scaling generator $G$, its Hamiltonian vector field $X_{H}$ is tangent to the level surfaces of $G$ and its flow commutes with the dilations generated by $X_{r}$. Since $i_{\left[X_{r},X_{H}\right]}\omega^{\prime}=\left[\mathcal{L}_{X_{r}},i_{X_{H}}\right]\omega^{\prime}$ and $\mathcal{L}_{X_{r}}\omega^{\prime}=r\omega^{\prime}$, one~has
\begin{align}
i_{\left[X_{r},X_{H}\right]}\omega^{\prime} & =\mathcal{L}_{X_{r}}i_{X_{H}}\omega^{\prime}-ri_{X_{H}}\omega^{\prime}\nonumber \\
 & =-d\left(X_{r}H\right)+rdH,
\end{align}
where we used the definition of $X_{H}$, $i_{X_{H}}\omega^{\prime}=-dH$. By~homogeneity of $H$, the~right-hand side vanishes, so $\left[X_{r},X_{H}\right]=0$ follows from the non-degeneracy of $\omega^{\prime}$. Consequently, the~symplectic flow of $X_{H}$ on the extended manifold $M$ projects onto the underlying contact manifold $C$ once the gauge is fixed. This projection reduces $X_{H}$ to a contact vector field $X_{h}$ on $C$, governed by a contact Hamiltonian {$h$} that is induced by the restriction of $H$ to the gauge slice. {It is straightforward to obtain the contact Hamiltonian $h$ by restricting $H$ from \eqref{H} to the gauge slice $Z=1$, which amounts to making the symplectization parameter $t=0$ in all expressions. As~a result, one finds that $h$ is identically zero on the contact manifold. In~other words, the~contact dynamics is trivial, as~one should expect for the scale flow in the scale-invariant thermodynamic description given by the contact~geometry. 



{Finally, we take the opportunity to distinguish the three main dynamical concepts we have introduced in this work: (i) gauge orbits, (ii) geometric flows, and~(iii) physical thermodynamic processes. Gauge orbits are generated by the canonical conformal field $X_r$. They represent continuous dilations of the extended coordinates and do not correspond to thermodynamic state changes. Instead, imposing a gauge-fixing condition (such as $Z=Z_0$) effectively takes the quotient of the symplectic manifold by these orbits, collapsing them to recover the physical scale-invariant contact geometry. Geometric flows are associated with the Hamiltonian vector field $X_G$ of the Noether charge $G$. This flow strictly preserves the symplectic volume and performs canonical scale transformations that inversely scale extensive and intensive variables to maintain the total energy invariant. It acts as a mathematical mechanism that ensures the proper scaling symmetries within the phase space, rather than representing a physical time evolution. The~physical thermodynamic processes are driven by the flow of $X_H$ in the symplectic space, which reduces to the contact vector field $X_h$ upon gauge fixing.}

\subsection{Schwarzschild Black~Hole}

Let us consider a case in which the thermodynamic description has only one independent extensive variable. By~extensive, we mean that the internal energy is homogeneous of degree $r$ and that there is a single coordinate that scales with weight $r_{q}$, for~arbitrary positive $r$ and $r_{q}$.~This simplest case illustrates all the previous developments and reveals some peculiar features. Starting from the intensive description, where all functions are {densities, 
 let us consider a three-dimensional contact manifold and contact form given in local Darboux coordinates as
\begin{equation}
\alpha=du-pdq.
\end{equation}

Let us assume that $q>0$. Since the thermodynamic equilibrium states are described by a homogeneous function $u\left(q\right)$ of degree zero, $u\left(\lambda^{r}q\right)=u\left(q\right)$ for $r>0$, then, by~the Euler theorem, $\frac{du}{dq}\left(q\right)=0$, that is, $u\left(q\right)$ is constant. Therefore, it also follows that $p=du/dq=0$ in the Legendre submanifold described by the potential $u\left(q\right)$,
\begin{equation}
u\left(q\right)={c}=\mathrm{const}.~,\ p\left(q\right)=0.
\end{equation}

Symplectization in this case produces the symplectic manifold $M\simeq\mathbb{R}^{4}$ with symplectic form $\omega=\left(e^{rt}\alpha\right)$. In~the coordinates given by the symplectomorphism (\ref{eq:transf-can}) we~have
\begin{equation}
Z=e^{\rho t}\,,\,\,\mu=\rho^{-1}e^{\left(r-\rho\right)t}\left(ru-r_{q}pq\right)\,,\,\,\,\,Q=e^{r_{q}t}q\,,\,\,P=e^{\left(r-r_{q}\right)t}p.
\end{equation}

Then the symplectic form is given by
\begin{equation}
\omega=dZ\wedge d\mu+dQ\wedge dP.
\end{equation}

The new internal energy is given by (\ref{energy}),
\begin{equation}
U=e^{rt}u={c}Z^{r/\rho}
\end{equation}
and scales as
\begin{equation}
U\left(\lambda^{\rho}Z\right)=\lambda^{r}U\left(Z\right).
\end{equation}

One also has from (\ref{eq:euler-relation})
\begin{equation}
rU=\rho Z\mu+r_{q}PQ\,.
\end{equation}

The Lagrange submanifold in the new coordinates is given by
\begin{equation}
P=0\,,\,\,\mu\left(Z\right)={c}\frac{r}{\rho}Z^{r/\rho-1},\,\,Q=Z^{r_{q}/\rho}q\,,\,\,rU=\rho Z\mu\,,\,\,dU=\mu dZ.
\end{equation}

We observe that, in~the case of a three-dimensional contact manifold, the~transformed variable $Q$ does not participate in the extensive description, since the internal energy $U$ is a function only of $Z$. So even though the symplectic description is four dimensional, and~uses the internal variables $Q$ and $P$, they do not appear in the physical description of the system, which is entirely described by the function $U\left(Z\right)$. Since this is a thermodynamic theory, we must have an associated temperature. Therefore, we assume that, as~an equation of state, $\mu\left(Z\right)$ must be identified with a temperature $T$ and the new variable $Z$ with the extensive entropy $S$ of the system. For~the equilibrium states,
\begin{equation}
dU=TdS~,\ rU=\rho TS~,\ T={c}\frac{r}{\rho}S^{r/\rho-1}.
\end{equation}

Note that for $r=\rho$, the~temperature $T$ is constant:
\begin{equation}
T\left(S\right)={c}\label{tcon}
\end{equation}
and we can say that there is no {change of state} in the thermodynamic description, since only isothermal processes are allowed. From~a physical point of view, this example corresponds to the unusual case where the heat $\delta=TdS$ is an exact 1-form. That is, a~thermodynamics that admits the existence of caloric. In~a usual thermodynamic system, where the entropy scales as the energy ($r=\rho$), the~above result (\ref{tcon}) shows that such a system only admits isothermal process. However, a~non-frozen thermodynamic system can be obtained if one {changes the energy and entropy scaling weights}. Probably the most well-known examples are given by black hole thermodynamics~\cite{baldiotti2017,dolan2011}. In~this case, the internal energy is associated with the Misner--Sharp mass~\cite{misner} and the temperature with the surface gravity of the event horizon. Therefore, {$\rho-r$} represents, geometrically, the~difference between the volume and area dimensions. That is, for~arbitrary $D+1$ dimensions spacetime,  {$\rho-r=1$}. For~the Schwarzschild black hole in $3+1$ dimensions, geometric arguments show the entropy varies proportionally to the square of the mass, i.e.,~$r=1$ and $\rho=2$,
\begin{equation}
U\left(S\right)={c}S^{1/2}~,\ U\left(S\right)=2TS.\label{smarr}
\end{equation}

In this geometric development, {the Euler relation expressed in} the second equation above follows from a Komar integral and is known as Smarr's formula~\cite{campos}. We speculate that other thermodynamic descriptions, probably associated with black holes in fractal geometries, may be constructed by considering other values of  {$\rho-r$}. For~example, considering the Schwarzschild black hole with the fractal horizon described in~\cite{gashti,Barrow}, {we have $\rho-r=\Delta$, $0\leq \Delta \leq 1$, and} the procedure leading to (\ref{smarr}) yields:
\begin{equation}
\left(D-3\right)U=\left(D-2-\Delta\right)TS~,\ 0\leq\Delta\leq1,\label{gsr}
\end{equation}
with $\Delta=0$ for a usual smooth spacetime structure and $\Delta=1$ for the maximum intricacy ``roughest'' black hole surface. In~this last case, the~horizon area behaves as a volume from an information content perspective. The~new expression (\ref{gsr}) represents a generalized Smarr relation that is hard to be obtained by using a Komar integral,~{since the original formulation of the Komar integral requires concepts of smooth differential geometry that conflict directly with the discontinuous and non-differentiable nature of fractals. The~quantity
$\Delta$ can affect thermodynamic quantities as the specific heat and alter physical characteristics of the black hole. However, such analysis falls outside the scope of the present work.} The above construction demonstrates that, in~addition to geometric arguments, there is also a fundamental reason why the entropy of a {Schwarzschild} black hole cannot be proportional to its internal energy. Otherwise, the~thermodynamic description cannot reproduce the expected physical behavior. This is because a Schwarzschild black hole with a constant temperature cannot describe important phenomena such as, for~example, Hawking evaporation. {We stress that this change in weights is mandatory only for the Schwarzschild black hole. If~the black hole has additional independent thermodynamic parameters (e.g., charge or angular momentum),~it is always possible to redefine the entropy such that it scales with the mass parameter. In~the symplectic extended space, the~pair $\left(Q,P\right)$ encodes the internal relative degrees of freedom, that is, the~``shape'' or ``composition'' of the system independent of its global scale Z. Because~a Schwarzschild black hole is uniquely characterized by a single macroscopic extensive variable (its entropy or area), it possesses no internal shape parameters or relative fractions that can be independently excited. Consequently, the~internal energy cannot depend on Q, which rigorously forces its conjugate generalized force to vanish ($P=\partial U/\partial Q=0$). This zero-momentum constraint reflects the physical impossibility of altering the internal thermodynamic state of a static, uncharged black hole without changing its global area.~For a Kerr or Reissner--Nordström black hole, however,~the~$\left(Q,P\right)$ sector would immediately cease to be trivial, physically mapping to the observable rotational or electromagnetic hair (e.g.,~the~dimensionless spin parameter and its conjugate angular velocity). Thus, while the decoupling of $\left(Q,P\right)$ in our model might initially appear as a formal mathematical artifact of the symplectization process, it can be seen as a geometric manifestation of the No-Hair~Theorem.}

\section{Conclusions}

In this work,~we developed a connection between the constrained Hamiltonian formalism with symplectic and contact geometries for the description of thermodynamics. Through the well-established process of contactization, we reviewed how thermodynamic equilibrium states, modeled as Lagrangian submanifolds, map naturally to Legendre submanifolds. These foundations were illustrated by explicitly demonstrating the diffeomorphism between the Lagrangian submanifolds of the ideal and van der Waals~gases. 

The primary contribution of this paper, however, lies in our symplectic approach to describing thermodynamic systems with general homogeneity. By~employing the process of symplectization, we expanded the thermodynamic phase space to include a global scale variable, which allowed us to formalize the dynamics of systems whose extensive and intensive quantities scale with arbitrary~weights. 

{To ground this formalism physically, it is important to clarify the interpretation of the global scale variable $Z$ introduced during the symplectization process. Physically, this variable corresponds to the absolute macroscopic size, or~total extensive magnitude, of~the thermodynamic system (analogous to the total number of particles, $N$, or~the total volume, $V$, depending on the chosen representation). In~standard thermodynamic equilibrium, the~intrinsic state of a homogeneous system is completely determined by scale-invariant intensive variables (such as temperature and pressure). These intensive variables are independent of the system's total size due to the strict Euler homogeneity of thermodynamic~potentials.}

{Consequently, $Z$ acts purely as a global extensive multiplier. Moving along the geometric flow generated by the scaling symmetry mathematically represents a continuous global dilation of the system, physically akin to scaling up a system by adding more of the identical substance. This operation scales all extensive quantities proportionally but leaves the intensive equilibrium state entirely invariant. Therefore, imposing a gauge-fixing condition (e.g., $Z = Z_0$) is the geometric equivalent of choosing an arbitrary macroscopic reference size (such as exactly one mole or one liter) to isolate and study the intrinsic, scale-invariant equations of state.}

We demonstrated that imposing a gauge-fixing condition on this continuous scale parameter---
by quotienting the symplectic manifold by the orbits of the associated dilation generator---we successfully recover the physical description of scale-invariant quantities on the underlying contact manifold. The~physical necessity of this generalized homogeneity framework is highlighted by our analysis of the Schwarzschild black hole. Under~our approach, we showed that  {the internal energy and entropy scaling weights must be changed} to permit non-isothermal~processes.

This framework for handling arbitrary scaling degrees provides an appropriate geometric basis for describing dynamic thermodynamic phenomena in non-standard systems, such as Hawking~evaporation.

A natural direction for future research lies in extending this geometric framework to the study of phase transitions and critical phenomena. As~a thermodynamic system approaches a continuous phase transition, the~standard Euler homogeneity of its thermodynamic potentials breaks down, giving way to generalized homogeneous functions, as described by the Widom scaling hypothesis~\cite{widom}. Because~our symplectization procedure explicitly incorporates arbitrary scaling weights through the extended global scale variable and the associated dilation generator, it provides a mathematical framework for modeling these critical scaling laws. Future work will investigate how the scale-invariant dynamics and gauge-fixing procedures developed here can be used to geometrically characterize critical points, exploring the signatures of phase transitions in both standard macroscopic systems and the thermodynamics of black holes~\cite{hawking}. {Finally, as~in classical mechanics, the~main application of the Dirac theory of constraints with respect to gauge theories is related to the process of quantization. Quantization of thermodynamics is briefly discussed in~\cite{brav1}, and~a canonical quantization was recently proposed in~\cite{baldiotti2026}. A~possible continuation of this work would be to consider the canonical quantization of the symplectic~approach.}


\vspace{6pt} 










\begin{thebibliography}{99}
	
\bibitem{brav1}
	Bravetti, A.  Contact geometry and thermodynamics. \emph{Int. J. Geom. Methods Mod. Phys.} \textbf{2019}, \emph{16}, 1940003.
	
	
\bibitem{baldiotti2016}
	Baldiotti, M.C.; Fresneda, R.; Molina, C. A Hamiltonian approach to Thermodynamics. \emph{Ann. Phys.} \textbf{2016}, \emph{373}, 245--256.
	
	
\bibitem{baldiotti2017}
	Baldiotti, M.C.; Fresneda, R.; Molina, C. A Hamiltonian approach for the Thermodynamics of AdS black holes. \emph{Ann. Phys.} \textbf{2017}, \emph{382}, 22--45.
	
	
\bibitem{balian2001}
	Balian, R.; Valentin, P. Hamiltonian structure of thermodynamics with gauge. \emph{Eur. Phys. J. B} \textbf{2001}, \emph{21}, 269--282.
	
	
\bibitem{bravetti2015}
	Bravetti, A.; Lopez-Monsalvo, C.S.; Nettel, F. Contact symmetries and Hamiltonian thermodynamics. \emph{Ann. Phys.} \textbf{2015}, \emph{361}, 377--400.
	
	
\bibitem{bravetti2015b}
	Bravetti, A.; Lopez-Monsalvo, C.; Nettel, F. Conformal Gauge Transformations in Thermodynamics. \emph{Entropy} \textbf{2015}, \emph{17}, 6150--6168.
	
	
\bibitem{quevedo2022}
	Aragon-Munoz, L.; Quevedo, H. Symplectic structure of equilibrium thermodynamics. \emph{Int. J. Geom. Methods Mod. Phys.} \textbf{2022}, \emph{19},~2250178.
	
	
\bibitem{bravetti2023}
	Bravetti, A.; Jackman, C.; Sloan, D. Scaling symmetries, contact reduction and Poincar{\'e}'s dream. \emph{J. Phys. A Math. Theor.} \textbf{2023}, \emph{56},~435203.
	
	
\bibitem{gosh2026}
	Ghosh, A.; Harikumar, E. Hamiltonian Thermodynamics on Symplectic Manifolds. \emph{Int. J. Theor. Phys.} \textbf{2026}, \emph{65}, 134.
	
	
\bibitem{abraham-marsden}
	Abraham, R.; Marsden, J. \emph{Foundations of Mechanics}, 2nd ed.; Benjamin/Cummings Publishing Company: Reading, MA, USA, 1978.
	
	
\bibitem{weinstein}
	Weinstein, A. Symplectic manifolds and their lagrangian submanifolds. \emph{Adv. Math.} \textbf{1971}, \emph{6}, 329--346.
	
	
\bibitem{cannas}
	Da Silva, A.C. \emph{Lectures on Symplectic Geometry}; Springer: Berlin/Heidelberg, Germany, 2008.
	
	
\bibitem{souriau}
	Souriau, J.-M.; Cushman-de-Vries, C.H.; Cushman, R.H.; Tuynman, G.M. \emph{Structure of Dynamical Systems: A Symplectic View of Physics}; Birkh{\"a}user: Boston, MA, USA, 1997.
	
	
\bibitem{dolan2011}
	Dolan, B.P. The cosmological constant and black-hole thermodynamic potentials. \emph{Class. Quantum Gravity} \textbf{2011}, \emph{28}, 125020.
	
	
\bibitem{misner}
	Misner, C.W.; Sharp, D.H. Relativistic equations for adiabatic, spherically symmetric gravitational collapse. \emph{Phys. Rev.} \textbf{164}, \emph{136}, B571--B576.
	
	
\bibitem{campos}
	Campos, T.L.; Baldiotti, M.C.; Molina, C. Generating Kerr-anti-de Sitter thermodynamics. \emph{Phys. Rev. D} \textbf{2024}, \emph{110}, 024049.
	
	
\bibitem{gashti}
	Gashti, S.N.; Pourhassan, B.; Sakall{\i}, {\.I}. Thermodynamic topology and phase space analysis of AdS black holes through non-extensive entropy perspectives. \emph{Eur. Phys. J. C} \textbf{2025}, \emph{85}, 305.
	
	
\bibitem{Barrow}
	Barrow, J.D. The Area of a Rough Black Hole. \emph{Phys. Lett. B} \textbf{2020}, \emph{808}, 135643.
	
	
\bibitem{widom}
	Widom, B. Equation of State in the Neighborhood of the Critical Point. \emph{J. Chem. Phys.} \textbf{1965}, \emph{43}, 3898--3905.
	
	
\bibitem{hawking}
	Hawking, S.W.; Page, D.N. Thermodynamics of black holes in anti-de Sitter space. \emph{Commun. Math. Phys.} \textbf{1983}, \emph{87}, 577--588.
	
	
\bibitem{baldiotti2026}
	Santos, L.F.; Ramos, V.H.M.; Cius, D.; Baldiotti, M.C.; Amaral, B. Canonical quantization for Equilibrium Thermodynamics. \emph{Phys. Rev. E} \textbf{2026}, \emph{113}, 05413.
	

\end{thebibliography}
\end{document}